\documentclass[a4paper,UKenglish,cleveref, autoref, thm-restate]{lipics-v2021}
\hideLIPIcs  %uncomment to remove references to LIPIcs series (logo, DOI, ...), e.g.  when preparing a pre-final version to be uploaded to arXiv or another public repository

\usepackage{mathcomp}
\usepackage{ascmac}
\usepackage{comment}
\newtheorem*{definition*}{Definition}

\newcommand{\cG}{3}
\newcommand{\cB}{2}
\newcommand{\cR}{1}
\newcommand{\cW}{0}

\newcommand{\rot}{\mathrm{rot}}
\newcommand{\splay}[1]{\mathsf{splay}(#1)}

\title{\texorpdfstring{A Refined Analysis of the Sequential Access Theorem for Splay Trees}{A Refined Analysis of the Sequential Access Theorem for Splay Trees}} %TODO Please add

\titlerunning{A Refined Analysis of the Sequential Access Theorem for Splay Trees} %TODO optional, please use if title is longer than one line

\author{Naonori Kakimura}{Department of Mathematics, Keio University, Japan}{kakimura@math.keio.ac.jp}{https://orcid. org/0000-0002-3918-3479}{Supported by JSPS KAKENHI Grant Numbers 23K21646, 26K02867, 26K21777 and JST ERATO Grant Number JPMJER2301, Japan.}%TODO %\author{Jane {Open Access}}{Dummy University Computing Laboratory, [optional: Address], Country \and My second affiliation, Country \and \url{http://www. myhomepage.edu} }{johnqpublic@dummyuni. org}{https://orcid. org/0000-0002-1825-0097}{(Optional) author-specific funding acknowledgements}%TODO mandatory, please use full name; only 1 author per \author macro; first two parameters are mandatory, other parameters can be empty.  Please provide at least the name of the affiliation and the country.  The full address is optional.  Use additional curly braces to indicate the correct name splitting when the last name consists of multiple name parts. 

\author{Yoshihiko Terai}{Japan}{}{}{}

\authorrunning{N. Kakimura and Y. Terai} %TODO mandatory.  First: Use abbreviated first/middle names.  Second (only in severe cases): Use first author plus 'et al. '

\Copyright{Naonori Kakimura and Yoshihiko Terai} %TODO mandatory, please use full first names.  LIPIcs license is "CC-BY";  http://creativecommons. org/licenses/by/3. 0/

 %TODO mandatory; please add comma-separated list of keywords

\category{} %optional, e.g. invited paper

\relatedversion{A preliminary version appears in the 37th International Symposium on Algorithms and Computation~(ISAAC 2026).} 
\nolinenumbers %uncomment to disable line numbering

\EventEditors{Lin Chen and Nicole Megow}
\EventNoEds{2}
\EventLongTitle{37th International Symposium on Algorithms and Computation (ISAAC 2026)}
\EventShortTitle{ISAAC 2026}
\EventAcronym{ISAAC}
\EventYear{2026}
\EventDate{December 6--9, 2026}
\EventLocation{Hangzhou, China}
\EventLogo{}
\SeriesVolume{399}
\ArticleNo{23}
\begin{document}
\maketitle

\begin{abstract}
A splay tree is a self-adjusting binary search tree that allows access, insertion, and deletion to be performed in amortized $O(\log n)$ time, where $n$ is the number of stored elements.
The sequential access theorem states that, when the elements of a splay tree are accessed in increasing order, the amortized cost per operation becomes a constant. 
In this paper, we show that the upper bound for this constant is at most $5.5$ by refining the existing analysis and introducing a new potential function.
Furthermore, we complement our result by showing that there exists a splay tree for which the constant is lower-bounded by almost $4$.
\end{abstract}

%\clearpage
\section{Introduction}

A binary search tree is a data structure that maintains a subset of a totally ordered set, allowing to \textsf{access}, \textsf{insert}, and \textsf{delete} an element.
The cost of a binary search tree is usually proportional to the height of the tree, and thus efficient binary search trees, such as
the red-black tree~\cite{RedBlackTree}, the AVL tree~\cite{AVLTree}, and the scapegoat tree~\cite{galperin1993scapegoat},  maintain their height to be $O(\log n)$ where $n$ is the number of elements.

A \textit{splay tree} is a self-adjusting binary search tree, invented by Sleator and Tarjan~\cite{sleator1985self}. 
A splay tree reshapes itself in response to each access. 
After accessing an element $x$, the tree moves $x$ to the root position by repeatedly performing so-called zig--zig, zig--zag, and zig operations, each of which is a sequence of rotations of trees.
Although their height is not guaranteed to be logarithmic in the number of elements, the total cost for processing $m$ requests is shown to be $O((m + n)\log n)$.
Thus, the worst case cost amortized over all the requests is the same as any balanced binary search tree.

    The sequential access theorem for splay trees states that, if we access elements of a splay tree in an ordered way, the total cost is reduced to $O(n)$. 
    In fact, Tarjan~\cite{tarjan1985sequential} proved that the total number of rotations is at most $10.8n$. 
    Elmasry~\cite{elmasry2004sequential} later analyzed\footnote{It is claimed in~\cite{elmasry2004sequential} that the upper bound is at most $4.5n$, but this is the number of ``links'' or zig--zig operations for \textit{semi-splaying}~(see the paragraph just before Theorem~2 in~\cite{elmasry2004sequential}). As each link takes $2$ rotations and we need additional $n$ rotations other than semi-splaying, their result would imply that the number of rotations is bounded by $9.5n$. See the end of Section~\ref{sec:mainproof} for details.} that the bound is at most $9.5n$.
    See also~\cite{Pettie10a,Sundar} for alternative proofs with worse bounds.

    The main contribution of this paper is to improve the upper bound in the sequential access theorem to $5.5n$. 
    To this end, we simplify the analysis of Elmasry~\cite{elmasry2004sequential} providing a more rigorous analysis.
    His analysis~\cite{elmasry2004sequential} introduces a coloring scheme for a splay tree, and defines a potential function based on the coloring. 
    In this paper, we modify their coloring function, and moreover, introduce a potential function so that the cost of splay operations near the root can be amortized.
    In addition, to complement with the main result, we observe in Section~\ref{sec:LB} that there exists a splay tree that takes at least $(4-o(1))n$ rotations in the sequential access.
    
\subsection{Other Related Work}

Splay trees are known to satisfy adaptive bounds, including the working‑set theorem~\cite{sleator1985self}, the static‑optimality theorem~\cite{sleator1985self}, and the dynamic‑finger theorem~\cite{Cole1,Cole2}.
A central open question on splay trees is the dynamic optimality conjecture~\cite{sleator1985self}, which states that splay trees achieve, up to a constant factor, the performance of the best possible dynamic binary search tree, even one that knows the whole access sequence in advance.

The sequential access theorem is a special case of a corollary of the dynamic optimality conjecture, called the traversal conjecture~\cite{sleator1985self}.
The traversal conjecture states that, if we access a splay tree $T$ in the order given by the preorder sequence of another tree $T'$, 
then the total cost is $O(n)$. 
The sequential access theorem is equivalent to the traversal conjecture when $T'$ is a skewed binary tree.
Another special case of the traversal conjecture when $T=T'$ was shown to be true in~Chaudhuri and H\"oft~\cite{ChaudhuriH93}.
Levy and Tarjan~\cite{LevyT19} proved the traversal conjecture when $T$ is $\alpha$-weight balanced.

Another candidate of a binary search tree satisfying dynamic optimality is \textsc{Greedy}~\cite{DemaineHIKP09,Lucas,Munro}.
\textsc{Greedy} is known to satisfy the sequential access theorem~\cite{ChalermsookG0MS15Greedy,Fox11}.
Moreover, the traversal conjecture holds if a linear-cost preprocessing is allowed~\cite{ChalermsookG0MS15}.
See also~\cite{Berendsohn0O24,ChalermsookPY24,Pettie08} and references therein for the analysis of binary search trees with pattern-avoiding inputs.

\section{Splay Trees}

%\subsection{Splay Tree}

Let $T$ be a binary tree with root $r$. 
        For a vertex $x$ of $T$, the \textit{left sub-tree of $x$}, denoted by $T_{\ell}(x)$, is the sub-tree of $T$ such that the root is the left child of $x$. 
        The \textit{right sub-tree of $x$}, denoted by $T_{r}(x)$, is defined similarly. 
For a vertex $x$ of $T$ with $x\neq r$, we denote by $p(x)$ its parent.

Let $U$ be a totally ordered set.
In this paper, we assume for simplicity that $U=\mathbb{Z}$.
        We say that a binary tree $T$ with root $r$ is a \textit{binary search tree} if each vertex $x$ of $T$ has a value $a_x\in U$ satisfying that, for any vertex $y$ of $T_{\ell}(x)$, we have $a_y < a_x$, and, for any vertex $y$ of $T_{r}(x)$, we have $a_x<a_y$.
A binary search tree allows the following operations for a given value $u\in U$:
\begin{enumerate}%[labelwidth=5em]
        \item \textsf{access} $u$. Determine whether $T$ has a vertex with value $u$ or not. 
        \item \textsf{insert} $u$. Add a new vertex with value $u$ to $T$. 
        \item \textsf{delete} $u$. Remove a vertex with value $u$ from $T$. 
    \end{enumerate}
In this paper, we often identify a vertex $x$ with its value $a_x$, and we also say that a vertex $x$ is \textit{accessed} when we invoke \textsf{access} $a_x$.

%%%%%%%%%%%%%%%%%%%%%%%%%%%%%%%%%%%%%%%%
A \textit{rotation} is an operation on a binary search tree that changes the structure without interfering with the order of the elements.
Formally, for an edge $(x,y)$ such that $y$ is the parent of $x$,
a \textsf{rotation at} $(x,y)$ moves the vertex $x$ up in the tree and its parent $y$ down.
%A rotation moves one vertex $x$ up in the tree and its parent $p(x)$ down. 
%Such an operation is denoted by \textsf{rotate} $x$.
%A \textit{right rotation} is a rotation when $x$ is the left child of $y$, and a \textit{left rotation} is one when $x$ is the right child of $y$.
For example, 
when $x$ is the left child of $y$, a \textsf{rotation at} $(x,y)$ moves the right sub-tree of $x$ to the left sub-tree of $y$, and moves the vertex $y$ to the right child of $x$.
%Specifically, in the right rotation at $(x,y)$, the right sub-tree of $x$ moves to the left sub-tree of $y$, and the vertex $y$ becomes the right child of $x$.
%In the left rotation, the left sub-tree of $x$ moves to the right sub-tree of $y$,  and the vertex $y$ becomes the left child of $x$.
See Figure~\ref{fig:zig_step} for an illustration.

\begin{figure}[ht]
                    \centering
                    \includegraphics[scale=0.25]{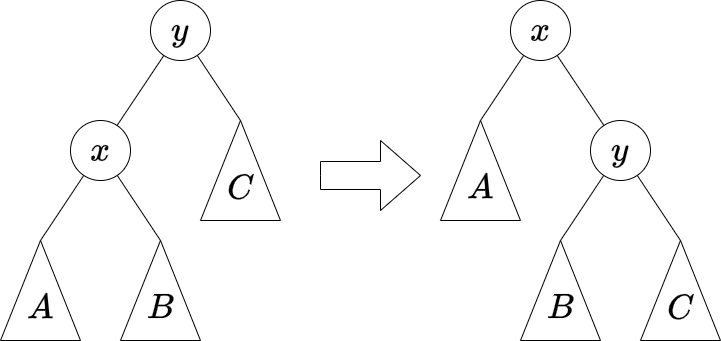}
                    \caption{A rotation, or a zig operation on $(x,y)$.}
                    \label{fig:zig_step}
                \end{figure}

    In a splay tree, each time a vertex $x$ of $T$ is accessed, a \textit{splay}, which is a sequence of rotations defined as below, is performed at $x$.

\begin{definition}
    Let $T$ be a binary search tree, and $x$ be a non-root vertex of $T$. 
    A \emph{splay at $x$} repeats the following steps until $x$ becomes the root of the tree. 
    We denote by $y$ the parent of $x$.
        \begin{itemize}
            \item \textbf{zig operation:} If $y$ is the root, then 
            invoke \textsf{rotation at} $(x,y)$~(Figure~\ref{fig:zig_step}), 
            which we call a \emph{zig operation on $(x,y)$}.
            \item \textbf{zig--zig operation:}
            If $x$ and $y$ are both left children~(or both right children), then first invoke \textsf{rotation at} the edge $(y, p(y))$ and then \textsf{rotation at} the edge $(x, y)$~(Figure~\ref{fig:zig_zig_step}). 
            We call it a \emph{zig--zig operation on} $(x, y, p(y))$. 
            \item \textbf{zig--zag operation:}
            If $x$ is a right~(left, resp.,) child and $y$ is a left~(right, resp.,) child, then first invoke \textsf{rotation at} the edge $(x,y)$, and then \textsf{rotation at} the resulting edge $(x,p(y))$~(Figure~\ref{fig:zig_zag_step}). 
            We call it a \emph{zig--zag operation on} $(x, y, p(y))$.
        \end{itemize}
        Moreover, invoking a splay operation at $x$ is denoted by $\splay{x}$. 
\end{definition}

In this paper, we omit the implementation details of \textsf{access}, \textsf{insert}, and \textsf{delete} for splay trees.
See e.g.,~\cite{sleator1985self} for details.

                           \begin{figure}[ht]
                    \centering
                    \includegraphics[scale=0.25]{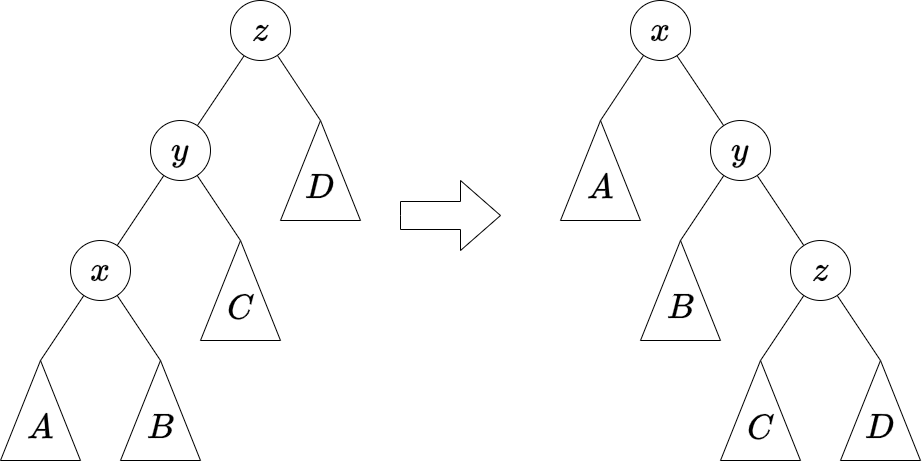}
                    \caption{A zig--zig operation on $(x,y,z)$.}
                    \label{fig:zig_zig_step}
                \end{figure}
                \begin{figure}[ht]
                    \centering
                    \includegraphics[scale=0.25]{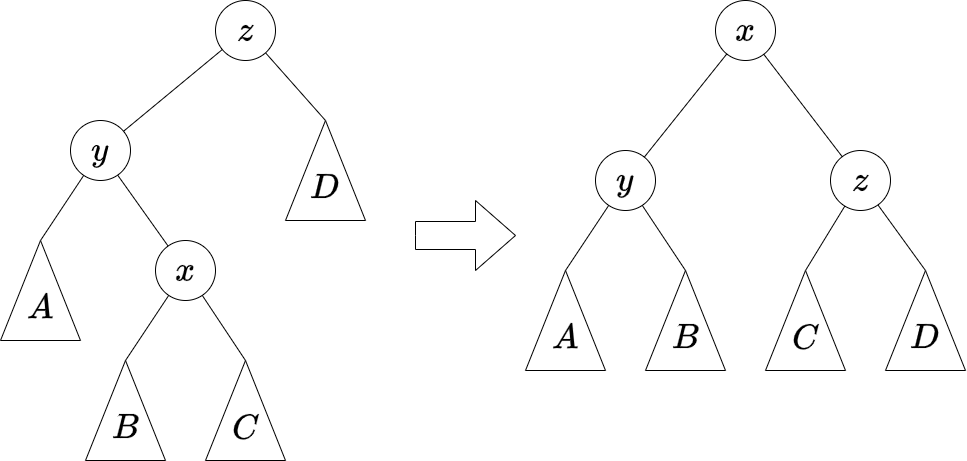}
                    \caption{A zig--zag operation on $(x,y,z)$.}
                    \label{fig:zig_zag_step}
                \end{figure}

\section{Main Theorem}\label{sec:main}

Let $T$ be a splay tree with a set $N\subseteq U$ of $n$ elements. 
In this section, we may simply assume that $N=\{1, 2, \dots, n\}$.
Moreover, we identify each vertex $v$ with its value $a_v$, that is, the vertex set of $T$ is identical to $N$. 

The \textit{sequential access for $T$} is to \textsf{access} the elements in the ascending order of $N$. 
That is, the sequential access invokes $\splay{t}$ for $t=1,2,\dots, n$ in this order.
We denote by $\rot(T)$ the number of rotations during the sequential access for $T$. 
The main theorem of this paper is the following.

\begin{theorem}\label{thm:main}
    For a splay tree $T$ with $n$ vertices, it holds that 
    $\rot(T)\leq 5.5 n$.
\end{theorem}

The rest of this section is devoted to proving Theorem~\ref{thm:main}.
We remark that, in Section~\ref{sec:LB}, we observe that there exists a splay tree $T$ such that $\rot(T)\geq (4-o(1))n$.

For a splay tree $T$, the \textit{left spine of a sub-tree} is defined to be the path from the root of this sub-tree to its leftmost leaf.  
Thus, every vertex on the path is the left child of its predecessor.  
The \textit{right spine} is defined analogously.
In particular, the right spine from the root is called the \textit{top spine}. 
Also, we call the left spine of the right sub-tree of the root the \textit{splaying spine}. 
See Figure~\ref{fig:coloring} for an example.

\subsection{Preliminaries}\label{sec:pre}

\paragraph*{Initialization}
To bound $\rot(T)$, we first slightly modify the given splay tree $T$ as follows:
We add a vertex with value $0$ to $T$ such that the vertex $0$ is the root and $T$ is the right sub-tree of the vertex $0$. 
Let $T_0$ be the resulting tree. 
Thus, $T_0$ is a splay tree with elements $0, 1, 2, \dots, n$.

Consider the sequential access for $T_0$, instead of $T$.
For each $t=1, 2,\dots, n$,
we denote by $T_t$ the splay tree just after invoking $\splay{t}$ to $T_{t-1}$.
We often call $T_t$ a \textit{splay tree at time $t$}.
Note that the leftmost leaf of $T_t$ is always the vertex $0$ for $t=1,2,\dots, n$.
It then follows that $T_1$, with the removal of the vertex $0$, is identical to the tree after invoking $\splay{1}$ to $T$.
Moreover, for any $t=1,2,\dots ,n-1$, $T_{t+1}$, with the removal of the vertex $0$, is identical to the tree after invoking $\splay{1}, \splay{2},\dots, \splay{t+1}$ to $T$ in this order.
Since the vertex $0$ is involved only once when splaying the vertex $1$, we have $\rot(T_0) = \rot(T) + 1$. 
Thus, in what follows, we aim to evaluate $\rot(T_0)$.

\paragraph*{Splaying at time $t$}

Let $T_t$ be a splay tree at time $t<n$.
We observe that the root of $T_t$ is $t$, and the vertex $t+1$ is the deepest vertex of the splaying spine, which we call the \textit{splaying vertex}, as the vertex $t+1$ will be splayed at time $t$.

Let $S_t$ be the set of vertices on the splaying spine of $T_t$, and let $s_t=|S_t|$. We denote by $y_1, y_2, \dots, y_{s_t}$ the vertices of $S_t$ in this order from the leaf, that is, $y_{i+1}=p_t(y_i)$ for $i=1,\dots, s_t-1$, where $p_t(x)$ denotes the parent of a vertex $x$ in $T_t$.
Then we have $y_{1} = t+1$ and $y_{s_t}$ is on the top spine.

When invoking $\splay{t+1}$ to $T_t$, we first perform a zig--zig operation on $(y_1, y_2, y_3)$~(if $s_t\geq 3$).
After that, we perform a zig--zig operation on $(y_1, y_4, y_5)$.
We repeat such zig--zig operations on $(y_1, y_\ell, y_{\ell+1})$~($\ell=2,4,6,\dots$) for the splaying spine until $y_1$ moves to the top spine.
If $s_t$ is odd, the last zig--zig operation is performed on $(y_1, y_{s_t-1}, y_{s_t})$, and we finally perform a zig operation on $(y_1, t)$ to move $y_1$ to the root.
If $s_t$ is even, then the last zig--zig operation is performed on $(y_1, y_{s_t-2}, y_{s_t-1})$, and finally we perform a zig--zag operation on $(y_1, y_{s_t}, t)$ to move $y_1$ to the root.
Therefore, since each of zig--zig and zig--zag operations takes $2$ rotations, the number of rotations for $\splay{t+1}$ is equal to $s_t$.

We say that a tuple $(z, w)$ of two vertices of $S_t$ is a \textit{pair at time $t$} if $w$ is the parent of $z$ in $T_t$ and $z$ becomes the parent of $w$ in $T_{t+1}$.
That is, a pair $(z, w)$ is two consecutive vertices on the splaying spine such that a zig--zig operation on $(y_1, z, w)$ is performed when splaying the vertex $t+1$.
We denote by $P_t$ the set of pairs at time $t$.

\subsection{Main Proof}\label{sec:mainproof}

To evaluate the number of rotations in the sequential access, we define a potential function based on coloring vertices.
We first introduce a coloring on vertices as follows.
See Figure~\ref{fig:coloring} for an illustration.

\begin{definition}\label{def:coloring}
At time $t$, we define a coloring function $c_t:N\to \{\cW, \cR, \cB, \cG\}$ for each vertex as follows. 
    \begin{enumerate}
        \item If $x$ is the root of $T_t$, then $c_t(x) = \cG$~\textup
        {(\textsf{green})}. 
        \item If $x$ is on the top spine of $T_t$ but not the root of $T_t$, then $c_t(x) = \cB$~\textup
        {(\textsf{black})}. 
        \item If $x$ is on the splaying spine of $T_t$, but not on the top spine, then $c_t(x) = \cR$~\textup
        {(\textsf{red})}. 
        \item Otherwise, $c_t(x) = \cW$~\textup
        {(\textsf{uncolored})} if $t=0$, and $c_t(x) = c_{t-1}(x)$ if $t\geq 1$. 
    \end{enumerate}
    A vertex $x$ with $c_t(x)\neq 0$ is called a \textup{colored} vertex.
\end{definition}

\begin{figure}[ht]
    \centering
    \includegraphics[scale=0.35]{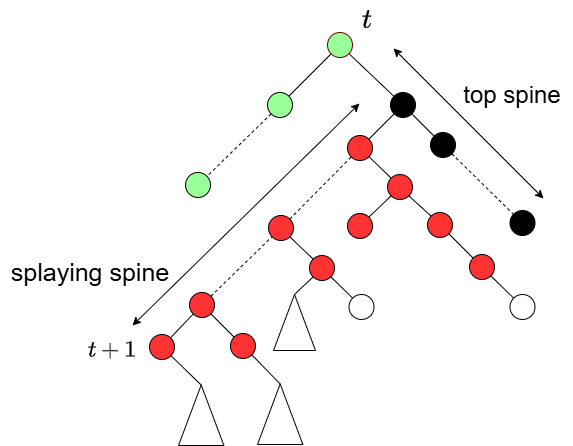}
    \caption{An illustration of coloring of $T_t$.}
    \label{fig:coloring}
\end{figure}

We remark that our coloring function is a simplified version of that of Elmasry~\cite{elmasry2004sequential}.
In~\cite{elmasry2004sequential}, the vertices with color $\cR$ are partitioned into two classes with distinct colors, and the left spine of the root has no colors, where the rotations related to the root are not counted in the analysis of~\cite{elmasry2004sequential}.

It is observed from Definition~\ref{def:coloring} that a vertex $x$ becomes colored once $x$ is on the splaying spine, and $x$ remains colored until the end.
We also see that the colored vertices induce a connected tree.

For a vertex $x$ of $T_t$, we define $g_t(x)$ by the number of vertices $y$ with $c_t(y)=1$ on the right spine of $x$, and $h_t(x)$ by the number of vertices $y$ with $c_t(y)=1$ on the right spine of the left child of $x$. 
We also define $v_t(x) = g_t(x) - h_t(x)$. 

The following lemma is observed in~\cite{elmasry2004sequential}.

\begin{lemma}[Elmasry~\cite{elmasry2004sequential}]\label{lem:delta_not_t}
    For any pair $(z, w)\in P_t$ at time $t < n$, the following hold: 
    \begin{enumerate}
        \item $g_{t+1}(w) = g_t(w)$.
        \item $h_{t+1}(w) = h_t(w) - 1$.
        \item $g_{t+1}(z) = g_t(w) + 1$.
        \item $h_{t+1}(z) \leq h_t(z) + 1$.
    \end{enumerate}
\end{lemma}

This implies the following corollary.

\begin{corollary}\label{cor:Elmasry_v}
    For any pair $(z, w)\in P_t$ at time $t < n$, the following hold: 
\begin{enumerate}%\label{eq:v_diff}
\item $v_{t+1}(w)  =  v_t(w)+1$.
\item $v_{t+1}(z)  \geq v_t(w)+v_t(z)$.
\end{enumerate}
\end{corollary}
\begin{proof}%[Proof of Corollary~\ref{cor:Elmasry_v}]
By Lemma~\ref{lem:delta_not_t}, we obtain
\begin{align*}
v_{t+1}(w) & = g_{t+1}(w) - h_{t+1}(w) = g_t(w) - \left(h_t(w) - 1\right) = v_t(w)+1,\\
v_{t+1}(z) & = g_{t+1}(z) - h_{t+1}(z) \geq \left(g_t(w) +1 \right)- \left(h_t(z) + 1\right) = v_t(w)+v_t(z),
\end{align*}
where the last equality of the second one follows as $h_t(w)=g_t(z)$.
\end{proof}

We summarize basic properties of the functions $g_t$, $h_t$, and $v_t$.
%See the full version~\cite{} for the proofs.

\begin{lemma}\label{lem:nonneg_v}
    Let $z, w$ be any two vertices of $T_t$ such that $c_t(z)=c_t(w)=1$ and 
    $z$ is a descendant of $w$.
    Then we have $g_t(z)\leq g_t(w)$. 
\end{lemma}
\begin{proof}
    We will show the lemma by induction on time $t$.
    When $t=0$, any vertex with color $1$ belongs to the splaying spine $S_0$ and $g_0(x) = 1$ for any $x\in S_0$ such that $x$ is not on the top spine.
    Hence the lemma holds. 

    Let $t\geq 1$, and assume that $g_{t-1}(z)\leq g_{t-1}(w)$ holds for any two vertices $z, w$ such that $z$ is a descendant of $w$ and $c_{t-1}(z)=c_{t-1}(w)=\cR$.
    It suffices to show the case when $z$ is a child of $w$, as applying this case repeatedly would imply the case when $z$ is a descendant of $w$.
    If $z$ is the right child of $w$, it holds by the definition of $g_t$ that $g_{t}(w)= g_{t}(z)+1$, and thus the lemma holds.
    Therefore, in what follows, we consider the case when $z$ is the left child of $w$.
    Then, it follows that either $w, z\not\in S_t$ or $w, z\in S_t$, where we recall that $S_t$ is the set of the vertices of the splaying spine of $T_t$.
    
    If $w, z\not\in S_t$, then we observe that $z$ was a descendant of $w$ at time $t-1$, and moreover, we have $g_{t}(w)=g_{t-1}(w)$ and $g_{t}(z)=g_{t-1}(z)$.
    Thus, the lemma holds by induction.
    
    Next consider the case when $w, z\in S_t$. 
    Then they satisfy one of the following three conditions: (a) $w, z\in S_{t-1}$, (b) $w\in S_{t-1}$, $z\not\in S_{t-1}$, or (c) $w, z\not\in S_t$.

    \begin{description}
        \item[(a) $z, w\in S_{t-1}$.]
         This means that, at time $t-1$, there exist two vertices $y, y'$ such that $(z,y)$ and $(w,y')$ are pairs and $w=p_{t-1}(y)$.
         Then we see by Lemma~\ref{lem:delta_not_t} that $g_t(z) = g_{t-1}(y)+1$ and $g_t(w) = g_{t-1}(y')+1$.
         Since $y$ is a descendant of $y'$ in $T_{t-1}$, it holds that $g_{t-1}(y) \leq g_{t-1}(y')$ by the induction hypothesis.
         Therefore, we have $g_t(z)\leq g_t(w)$. 
        \item[(b) $z\not\in S_{t-1}$, $w\in S_{t-1}$.]
        Then, at time $t-1$, $z$ is on the left spine of the right sub-tree of $t$ in $T_{t-1}$.
        It follows that, if $z$ is colored at time $t-1$, $g_t(z) = g_{t-1}(z)$, and otherwise, we have $g_{t-1}(z)=0$ and $g_t(z)=1$ as $z\in S_t$.
        Thus, $g_t(z) = \max\{1, g_{t-1}(z)\}\leq g_{t-1}(z)+1$.
        On the other hand, 
        since $w\in S_{t-1}$, there exists a vertex $y'$ such that $(w,y')$ is a pair at time $t-1$.
        By Lemma~\ref{lem:delta_not_t}, $g_t(w) = g_{t-1}(y')+1$.
        Since $g_{t-1}(w) \leq g_{t-1}(y')$ by the induction hypothesis, we obtain $g_t(w) \geq g_{t-1}(w)+1$.
        We observe that $g_{t-1}(w) \geq g_{t-1}(z)$, since the inequality clearly holds if $g_{t-1}(z)=0$ and by induction if $g_{t-1}(z)\geq 1$. 
        Hence, we obtain $g_t(w) \geq g_{t-1}(z)+1\geq g_t(z)$. 
        \item[(c) $z, w\not\in S_{t-1}$.]%\leavevmode\par
       Then, at time $t-1$, $z, w$ are both on the left spine of the right sub-tree of $t$ in $T_{t-1}$.
       Similarly to the previous case, it follows that $g_t(z) = \max\{1, g_{t-1}(z)\}$ and $g_t(w) = \max\{1, g_{t-1}(w)\}$.
       Then $g_t(z)\leq g_t(w)$ clearly holds if $g_{t-1}(z)=0$.
       Otherwise, that is, if $z$ is colored at time $t-1$, then $w$ is also colored at time $t-1$, and hence we have $g_{t-1}(z) \leq g_{t-1}(w)$ by induction.
       Hence, we obtain $g_t(z)\leq g_t(w)$.
        \end{description}
        Therefore, in each case, we have shown that $g_t(z) \leq g_t(w)$.
\end{proof}

\begin{corollary}\label{cor:nonnegative_v}
    For any vertex $x$ with $c_t(x)=1$ at time $t$, it holds that $v_t(x)\geq 0$.
\end{corollary}
\begin{proof}%[Proof of Corollary~\ref{cor:nonnegative_v}]
If $x$ has no left child, then $h_t(x) = 0$, and hence the claim holds as $g_t(z)\geq 0$.
Otherwise, letting $w$ be the left child of $x$, we have $v_t(x) = g_t(x) - h_t(x) = g_t(x) - g_t(w) \geq 0$ by Lemma~\ref{lem:nonneg_v}.
Thus, the corollary holds.
%\qed
\end{proof}

%\subsection{Main Proof}

We next define a function, called a \textit{potential}, using the coloring function $c_t$.

\begin{definition}\label{def:potential}
For any time $t\leq n$, we define a \emph{potential} $\Phi_t: N\to \mathbb{R}_{\geq 0}$ as follows. 
    \begin{itemize}
        \item If $c_t(x) = \cW$, then $\Phi_t(x) = 5. 5$. 
        \item If $c_t(x) = \cR$, then \[
        \Phi_t(x) = \frac{h_t(x)^2}{2} + \alpha(v_t(x)),
        \quad \text{where} \quad
    \alpha(n) = 
        \begin{cases}
            5  & \textup{(if $n = 0$)} \\
            \max\{3-n, 0\}  & \textup{(if $n \geq 1$)}
        \end{cases}.        
        \]
      \item If $c_t(x) = \cB$, then  
     \begin{equation}
          \Phi_t(x)=
          \begin{cases}
           1                       & \textup{(if $0\leq h_t(x) \leq 1$)} \\
            2h_t(x) - 2 & \textup{(if $h_t(x)\geq 2$)}
          \end{cases}. 
        \end{equation}
        \item If $c_t(x) = \cG$, then $\Phi_t(x) = 0$. 
    \end{itemize}
    Moreover, a \emph{potential} $\Phi(T_t)$ of a splay tree $T_t$ at time $t$ is defined by 
    \[
    \Phi(T_t) = \sum_{x\in N}\Phi_t(x). 
    \]
\end{definition}

The main part of the proof of Theorem~\ref{thm:main} is to show the following.
Let $r_t$ be the number of rotations when invoking $\splay{t}$ to $T_{t-1}$ for $t=1,2,\dots, n$.
Then we have $\rot (T_0)=\sum_{t=1}^n r_t$.

\begin{lemma}\label{lem:main}
    For any time $t=0, 1, 2, \dots, n-1$, it holds that
    \begin{equation*}\label{main_ineq}
        r_{t+1} \leq \Phi(T_t) - \Phi(T_{t+1}).
    \end{equation*}
\end{lemma}

Assuming that the above lemma is true, Theorem~\ref{thm:main} easily follows as below.

\begin{proof}[Proof of Theorem~\ref{thm:main}]
First consider when $t=0$.
Then all the vertices $x$ with $c_0(x)=1$ are on the splaying spine, which implies that $h_0 (x)\leq 1$ for any vertex $x$ with $c_0(x)=1$, and hence $\Phi_0(x)=\frac{h_t(x)^2}{2} + \alpha(v_t(x))\leq 5.5$.
If $c_0(x)=2$, then $x$ is on the top spine, and hence $\Phi_0(x) =1$ as $h_0 (x)\leq 1$.
We also have $\Phi_0(x)= 5.5$ for an uncolored vertex $x$. 
Hence, $\Phi_0 (x)\leq 5.5$ for any vertex $x\in N$, which implies that 
\[
\Phi(T_0) = \sum_{x \in N}\Phi_0(x) \leq 5.5n. 
\]
Therefore, it follows from Lemma~\ref{lem:main} that  
    \begin{equation*}\label{pot_eq}
\rot(T_0) = \sum_{t=1}^{n}r_t \leq \sum_{t=0}^{n-1}\left(\Phi(T_{t})-\Phi(T_{t+1}) \right)
\leq \Phi(T_0) - \Phi(T_n)
\leq 5.5n,
\end{equation*}
where we note that $\Phi (T_n)\geq 0$ since $\Phi_t (x)\geq 0$ for any vertex $x\in N$.
Thus, we obtain $\rot (T)\leq \rot (T_0)\leq 5.5n$, which proves Theorem~\ref{thm:main}.
\end{proof}

To show Lemma~\ref{lem:main}, we further classify the set $P_t$ of pairs at time $t$ to evaluate their potential.
A pair $(z, w) \in P_t$ is said to be the \textit{top pair} if $w$ is on the top spine, i.e., $c_t(w)=2$.
A pair $(z, w) \in P_t$ is the \textit{bottom pair} if $(z, w)$ is not the top pair and $z$ is the parent of the splaying vertex $t+1$.
If a pair is neither the top pair nor the bottom pair, the pair is called a \textit{middle} pair.
We denote by $M_t$ the set of the middle pairs at time $t$.
We also denote by $O_t$ the set of vertices not in $M_t$, that is, $O_t =N\setminus \bigcup_{(w,z)\in M_t} \{w, z\}$.

We will prove the following two lemmas in the subsequent subsections.
\begin{lemma}\label{lem:red_rotation}
    For any time $t<n$, it holds that
    \[
%    \begin{equation}\label{red_rot_ineq}
        2|M_t| \leq \sum_{(z, w)\in M_t}\left(\Phi_t(z)+\Phi_t(w)-\Phi_{t+1}(z)-\Phi_{t+1}(w)\right).
    \]
%    \end{equation}
\end{lemma}

\begin{lemma}\label{lem:black_rotation}
    For any time $t<n$, it holds that
    \[
%    \begin{equation}\label{black_rot_ineq}
        r_{t+1} - 2|M_t|\leq \sum_{x\in O_t}\left(\Phi_t(x) - \Phi_{t+1}(x)\right).
%    \end{equation}
\]
\end{lemma}

The above two lemmas immediately imply Lemma~\ref{lem:main}, since taking the summations of the two inequalities yields
\[
r_{t+1} \leq \sum_{x\in N}(\Phi_t(x) - \Phi_{t+1}(x))
= \Phi(T_{t})-\Phi(T_{t+1}).
\]

\paragraph*{Comparision with the existing analysis}
We conclude this section by comparing our analysis with the analysis by Elmasry~\cite{elmasry2004sequential}.
Elmasry defined a potential function as $\frac{h_t(x)^2}{2}$ only for a vertex $x$ with $c_t(x)=1$, using a bit more involved coloring function.
Then he used the amount of decrease of this potential function to pay the cost of rotations.

Our analysis for the middle pairs~(Lemma~\ref{lem:red_rotation}) is essentially the same as the one by Elmasry~\cite{elmasry2004sequential}. 
He showed Lemma~\ref{lem:red_p} below for his potential function, and proved that the total number of zig--zig operations is at most $3.5n$.
This implies that the number of rotations involving middle and bottom pairs is at most $6.5n$.
In our analysis, 
we modify his potential function to the one in Definition~\ref{def:potential} to reduce $6.5n$ to $5.5n$.
This can be done by Lemma~\ref{lem:red_p}, combined with a careful analysis of the behavior of $v_t(x)$ when $v_t(x)$ is small.
See Lemma~\ref{lem:red_q} in Section~\ref{sec:proofred} for the details.

For the number of rotations involving the top pairs~(Lemma~\ref{lem:black_rotation}), 
Elmasry used the simple observation that the total number of top pairs in the sequential access is at most $n$.
Since each top pair takes at most $3$ rotations~(including a zig operation to the root when $s_t$ is odd), the total number of rotations can be upper-bounded by $6.5n+3n=9.5n$.
In our analysis, 
we introduce a potential function for vertices on the top spine~(i.e., vertices $x$ with $c_t(x)=2$) as in Definition~\ref{def:potential},
and prove that, in each iteration, our new potential function decreases enough to pay the cost of rotations involving the top pairs.
In fact, in Section~\ref{sec:proofblack}, we show by case analysis that the amount of decrease of our potential function at a bottom pair and a splaying vertex can be used to pay the cost of rotations.
Thus, we do not need to count the number of rotations involving the top pairs, and hence the total number of rotations is bounded by $5.5n$.

\subsection{Proof of Lemma~\ref{lem:red_rotation}}\label{sec:proofred}

We first revisit Lemma~\ref{lem:delta_not_t}~(4) as below, which will also be used in the next subsection.
We note that, if $(z, w)$ is a pair  at time $t < n$, then the vertex $z$ has a left child $y$, and, since $y$ is on the splaying spine, $h_t(z)=g_t(y)\geq 1$.
%See the full version~\cite{} for the proofs.

\begin{lemma}\label{lem:refined_h}
Let  $(z, w)\in P_t$ be a pair  at time $t < n$, and $y$ be the left child of $z$. 
\begin{itemize}
\item If $y$ is not the splaying vertex $t+1$, then $h_{t+1}(z) = h_t(z) + 1$.
\item Suppose that $y=t+1$.
\begin{itemize}
\item If $t+1$ has no right child, then $h_{t+1}(z) = h_t(z)-1$.
\item If $t+1$ has a right child and $g_t(t+1)=1$, then $h_{t+1}(z) = h_t(z)=1$.
\item If $g_t(t+1)\geq 2$, then $h_{t+1}(z) = h_t(z)-1$.
\end{itemize}
\end{itemize}
\end{lemma}
\begin{proof}
Suppose that the left child $y$ of $z$ is not $t+1$. 
Then $y$ belongs to some pair $(z', y)$ at time $t$, and 
$h_{t+1}(z)= g_{t+1}(z')$.
By Lemma~\ref{lem:delta_not_t}, we have $g_{t+1}(z')=g_t(y)+1=h_t(z)+1$.
Thus, the first statement holds.

Suppose that $y=t+1$.
If $t+1$ has no right child, then $h_t(z)=g_t(y)=1$,
 and $z$ has no left child after splaying $t+1$.
Hence, $h_{t+1}(z)=0$, implying that $h_{t+1}(z) = h_t(z)-1$.
Next assume that $t+1$ has a right child $y'$.
Then $y'$ becomes a left child of $z$ after splaying $t+1$.
If $g_t(y)\geq 2$, then $g_t(y')=g_t(y)-1\geq 1$, and hence $h_{t+1}(z)=g_{t+1}(y')=g_{t}(y')=h_t(z)-1$.
Otherwise, i.e., if $g_t(y)=h_t(z)=1$, then $y'$ has no color at time $t$, and $y'$ becomes colored $1$ after splaying $t+1$.
Hence, $h_{t+1}(z)=1$.
Thus the lemma holds.
%\qed
\end{proof}

   To show Lemma~\ref{lem:red_rotation}, it suffices to prove that, for any middle pair $(z, w)\in M_t$,
\begin{equation}\label{red_pair_pot_diff1}
        \Phi_t(z)+\Phi_t(w)-\Phi_{t+1}(z)-\Phi_{t+1}(w) \geq 2.
    \end{equation}
    We evaluate the left-hand side of~\eqref{red_pair_pot_diff1} by partitioning the potential $\Phi_t$ into two parts.
    Namely, we define
    \begin{align*}    %\begin{equation}\label{delta_p_delta_q}
     %   \begin{aligned}
            \Delta P(w,z) &=\frac{h_t(z)^2}{2}+\frac{h_t(w)^2}{2}    - 
\frac{h_{t+1}(z)^2}{2}-\frac{h_{t+1}(w)^2}{2} \quad \text{and}\\
            \Delta Q(w,z) &= 
    \alpha (v_t(z))+\alpha (v_t(w)) - \alpha (v_{t+1}(z)) - \alpha (v_{t+1}(w)).
%        \end{aligned}
%    \end{equation}
    \end{align*}    
    Then \eqref{red_pair_pot_diff1} is equivalent to    \begin{equation}\label{red_pair_pot_diff2}
        \Delta P(w,z) + \Delta Q(w,z) \geq 2.
    \end{equation}

\begin{lemma}\label{lem:red_p}
%[\cite{elmasry2004sequential}]
    For any middle pair $(z, w)\in M_t$ at time $t < n$, it holds that 
%    \begin{equation*}
$\Delta P(w,z) =     v_t(z) - 1$.
%    \end{equation*}
\end{lemma}

\begin{proof}
    Since $(z,w)$ is a middle pair, the left child of $z$ is not $t+1$.
    Hence, by Lemmas~\ref{lem:delta_not_t} and~\ref{lem:refined_h}, we have $h_{t+1}(w) = h_t(w) - 1$ and $h_{t+1}(z) = h_t(z) + 1$, which implies that
    \begin{align*}
    \Delta P(w,z) &=\frac{1}{2}\left(h_t(z)^2+h_t(w)^2- 
(h_t(z) + 1)^2-(h_t(w) - 1)^2\right)\\
& = h_t(w)-h_t(z) - 1 = v_t(z) - 1,
    \end{align*}
    as $h_t(w)=g_t(z)$.
\end{proof}

Since $\alpha$ is a decreasing function, 
we have $\alpha (v_t(z))\geq \alpha (v_{t+1}(z))$ and $\alpha (v_t(w))\geq \alpha (v_{t+1}(w))$ by Corollary~\ref{cor:Elmasry_v}, and hence $\Delta Q(w,z)\geq 0$.
    Moreover, we have the following lemma, which holds for any (not necessarily middle) pair.

\begin{lemma}\label{lem:red_q}
%[\cite{elmasry2004sequential}]
    For any pair $(z, w)\in P_t$ at time $t < n$, it holds that 
%    \begin{equation*}
$\Delta Q(w,z) \geq 3 - v_t(z)$.
%    \end{equation*}
\end{lemma}
\begin{proof}
    We prove the lemma by case analysis.
    To simplify the notation, we denote $\Delta P = \Delta P(w,z)$ and $\Delta Q = \Delta Q(w,z)$ in the proof.
    \begin{description}
        \item[1. When $v_t(z) \geq 3$]%\leavevmode\par
        %\noindent
%        \subparagraph*{1. When $v_t(z) \geq 3$.}
        Since $\Delta Q \geq 0\geq 3- v_t(z)$, the lemma holds.
        
        \item [2. When $1\leq v_t(z)\leq 2$ and $1\leq v_t(w) \leq 2$]%\leavevmode\par
        %\noindent
        %\subparagraph*{2. When $1\leq v_t(z)\leq 2$ and $1\leq v_t(w) \leq 2$.}
        Since $v_{t+1}(z)-v_t(z)\geq v_t(w)\geq 1$ by Corollary~\ref{cor:Elmasry_v}, it holds that $\alpha(v_t(z))-\alpha(v_{t+1}(z))\geq 1$.
        Similarly, we obtain $\alpha(v_t(w))-\alpha(v_{t+1}(w))\geq 1$, as $v_{t+1}(w)-v_t(w)=1$ by Corollary~\ref{cor:Elmasry_v}.
        Hence $\Delta Q \geq 2$ holds, which implies $\Delta Q \geq 2\geq 3 - v_t(z)$.
        
        \item [3. When $1\leq v_t(z) \leq 2$ and $v_t(w) \geq 3$]%\leavevmode\par
        %\noindent
        %\subparagraph*{3. When $1\leq v_t(z) \leq 2$ and $v_t(w) \geq 3$.}
        Since $1\leq v_t(z) \leq 2$, we have $\alpha(v_t(z)) = 3-v_t(z)$.
        Since $v_t(z) \geq 0$ by Corollary~\ref{cor:nonnegative_v}, it holds by Lemma~\ref{lem:delta_not_t} that $v_{t+1}(z)\geq v_t (w)\geq 3$, and hence $\alpha(v_{t+1}(z)) = 0$.
        Therefore, we see $\Delta Q \geq \alpha(v_t(z))-\alpha(v_{t+1}(z)) \geq 3-v_t(z)$.
        
        \item [4. When $v_t(z) = 0$ or $v_t(w) = 0$]%\leavevmode\par
        %\noindent
        %\subparagraph*{4. When $v_t(z) = 0$ or $v_t(w) = 0$.}
        When $v_t(w) = 0$, we have $v_{t+1}(w) =1$ by Lemma~\ref{lem:delta_not_t}.
        Hence $\alpha(v_{t+1}(w))-\alpha(v_t(w))\geq 3$.
        On the other hand, when $v_t(z) = 0$ but $v_t(w) \geq 1$, we see $v_{t+1}(z) \geq v_t(w) \geq 1$ by Lemma~\ref{lem:delta_not_t}, and hence $\alpha(v_{t+1}(v))-\alpha(v_t(v))\geq 3$.
        Therefore, we have $\Delta Q \geq 3\geq 3-v_t(z)$ in each case.
    \end{description}
\end{proof}

The inequality~\eqref{red_pair_pot_diff2} easily follows from Lemmas~\ref{lem:red_p} and~\ref{lem:red_q}, which implies Lemma~\ref{lem:red_rotation}.
\qed

\subsection{Proof of Lemma~\ref{lem:black_rotation}}\label{sec:proofblack}

We first analyze the bottom pair $(z,w)$~(assuming it exists).
We define $\Delta \Phi(x) = \Phi_{t}(x)-\Phi_{t+1}(x)$ for a vertex $x\in N$.

\begin{corollary}\label{cor:bottom_poT_rem}
        Suppose that there exists the bottom pair $(z, w)$ at time $t<n$.
\begin{itemize}
\item If $g_t(t+1)=1$ and the vertex $t+1$ has a right child, then 
\[
\Delta \Phi(w) + \Delta \Phi(z) \geq h_t(z)+\frac{5}{2}.
\]
\item Otherwise, that is, if either $g_t(t+1)\geq 2$ or $g_t(t+1)=1$ but $t+1$ has no right child, then 
\[
\Delta \Phi(w) + \Delta \Phi(z) \geq 2h_t(z)+2.
\]
\end{itemize}
\end{corollary}
\begin{proof}
     Suppose that $g_t(t+1)=1$ and $t+1$ has a right child.
     By Lemmas~\ref{lem:delta_not_t} and~\ref{lem:refined_h}, it holds that $h_{t+1}(z) = h_t(z)$ and $h_{t+1}(w) = h_t(w) - 1$.
    Hence, we have 
    \begin{align*}
    \frac{1}{2}\left(h_t(z)^2+h_t(w)^2- 
h_{t+1}(z)^2-h_{t+1}(w)^2\right) &= \frac{1}{2}\left(h_t(z)^2+h_t(w)^2- 
h_t(z)^2-(h_t(w) - 1)^2\right)\\
& = h_t(w)- \frac{1}{2} = g_t(z)- \frac{1}{2},
    \end{align*}
    as $h_t(w)=g_t(z)$.
    Moreover, Lemma~\ref{lem:red_q} says that
    \begin{equation}\label{eq:deltaQ}
\alpha (v_t(z))+\alpha (v_t(w)) - \alpha (v_{t+1}(z)) - \alpha (v_{t+1}(w))     \geq 3 - v_t(z).
    \end{equation}
    Therefore, we obtain
\[
\Delta \Phi(w) + \Delta \Phi(z) \geq \left(g_t(z)- \frac{1}{2}\right) + \left(3 - v_t(z)\right) = h_t(z)+\frac{5}{2}.
\]

    Suppose that either $g_t(t+1)\geq 2$ or $g_t(t+1)=1$ but $t+1$ has no right child.
    Then we have $h_{t+1}(z) = h_t(z)-1$  by Lemma~\ref{lem:refined_h}.
    It follows that 
    \begin{align*}
    &\frac{1}{2}\left(h_t(z)^2+h_t(w)^2- 
h_{t+1}(z)^2-h_{t+1}(w)^2\right)\\ &= \frac{1}{2}\left(h_t(z)^2+h_t(w)^2- 
(h_t(z)-1)^2-(h_t(w) - 1)^2\right)
= h_t(z)+h_t(w)-1.
    \end{align*}
    Therefore, using~\eqref{eq:deltaQ}, we obtain
\[
\Delta \Phi(w) + \Delta \Phi(z) \geq \left(h_t(z)+h_t(w)-1\right)+ \left(3 - v_t(z)\right) = 2h_t(z)+2,
\]
as $v_t(z)=g_t(z)-h_t(z)$ and $h_t(w)=g_t(z)$.
\end{proof}

For a splay tree $T_t$ at time $t<n$, we denote the vertices on the top spine of $T_t$ by $b_0, b_1, b_2, \ldots, b_\ell$ where $b_i$ is the parent of $b_{i+1}$ for $i=1,\dots ,\ell -1$.
Thus, $b_0$ is the root of $T_t$.
Recall that $s_t$ denotes the number of vertices on the splaying spine of $T_t$.

\begin{lemma}\label{lem:black_rotation_odd}
    For any time $t<n$ such that $s_t$ is odd, it holds that
    \[
        \sum_{x\in O_t}(\Phi_t(x) - \Phi_{t+1}(x)) \geq
        \begin{cases}
        1 & (\text{if\ }s_t=1)\\
        3 & (\text{if\ } s_t=3)\\
        5 & (\text{if\ } s_t\geq 5)
        \end{cases}.
    \]
\end{lemma}
%\begin{proof}[Proof of Lemma~\ref{black_rotation}]
\begin{proof}
We note that, if a vertex $x\in O_t$ is not contained in the top or bottom pair, then $\Delta \Phi (x)\geq 0$ holds, since $x$ is not involved in splaying if $x\neq t+1$ and $\Delta \Phi (t+1)\geq 0$.

%        Due to the space limitation, we omit the proof for the case when $s_t=1$.
%    See the full version~\cite{} for the details.

    \subparagraph*{When $s_t=1$~(Figure~\ref{fig:b1=t}).}
    %\leavevmode\par
We see that $b_1=t+1$ and, after splaying $t+1$, the left spine of $b_2$ becomes the splaying spine.
Since  $\Delta \Phi(b_1) = 1$ and $\Delta \Phi(x) \geq 0$ for any $x\in O_t$, we have $\sum_{x\in O_t}(\Phi_t(x) - \Phi_{t+1}(x)) \geq \Delta \Phi(b_1)=1$.

    \begin{figure}[ht]
        \centering
        \includegraphics[scale=0.3]{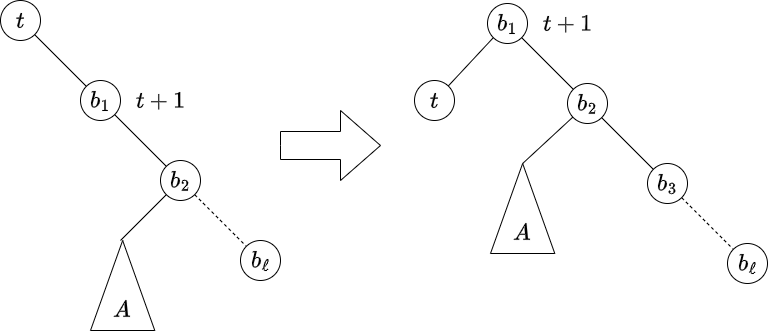}
        \caption{When $s_t = 1$.}
        \label{fig:b1=t}
    \end{figure}

%\begin{description}

\subparagraph*{When $s_t=3$~(Figure~\ref{c_and_t}).}
    \begin{figure}[ht]
        \centering
        \includegraphics[scale=0.3]{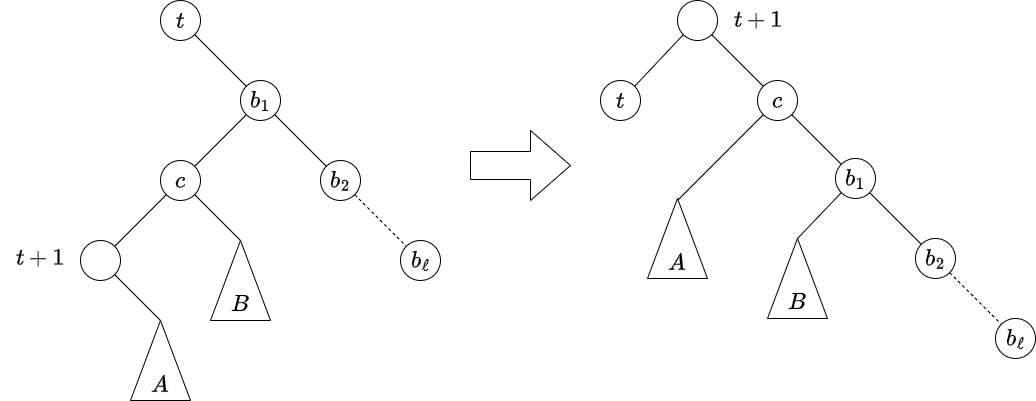}
        \caption{When $s_t=3$.}
        \label{c_and_t}
    \end{figure}
    There exists a top pair, denoted by $(c, b_1)$.
    After splaying $t+1$, the vertex $c$ moves to the top spine.
    Also, since $s_t=3$, there exists no bottom pair.

Since $\Delta\Phi(x) \geq 0$ for any $x \in O_t$ with $x\not\in\{b_1, c\}$, we have 
$\sum_{x\in O_t}(\Phi_t(x) - \Phi_{t+1}(x))\geq \Delta \Phi (b_1)+\Delta \Phi (c)+ \Delta \Phi (t+1)$.
Since $h_t(b_1)=g_t(c)$ and $h_{t+1}(b_1)=g_t(c)-1$ by Lemma~\ref{lem:delta_not_t}, we observe from the definition of $\Phi_t$ that 
\begin{equation}\label{eq:delta_b1}
\Delta \Phi (b_1)
=
\begin{cases}
0 & (g_t(c)=1)\\
1 & (g_t(c)=2)\\
2 & (g_t(c)\geq 3)
\end{cases}.
\end{equation}

Let $k=g_t(t+1)\geq 1$.
We will show $\Delta \Phi (b_1)+\Delta \Phi (c)+ \Delta \Phi (t+1)\geq 3$ by case analysis based on the value of $k$.
We note that $h_t(c)=k$ and also $g_t(c)\geq k$ by Lemma~\ref{lem:nonneg_v}.
\begin{enumerate}
\item Consider the case when $k\geq 3$.
Then, we see that $h_{t+1}(c)=k-1$ by Lemma~\ref{lem:refined_h}.
Since $\Phi_t (c)\geq k^2/2$ and $\Phi_{t+1} (c) = 2(k-1) -2$, we have
$\Delta \Phi (c) \geq k^2/2 -  2 k +4 \geq 5/2$ for $k\geq 3$.
Moreover, since $g_t(c)\geq k\geq 3$, we have $\Delta \Phi (b_1)= 2$ by~\eqref{eq:delta_b1}.
Hence, since $\Delta \Phi (t+1)=0$, it holds that $\Delta \Phi (b_1)+\Delta \Phi (c)+ \Delta \Phi (t+1) \geq  5/2+2\geq 3$.

\item Consider the case when $k=2$.
Then, since $\Phi_t (c)\geq k^2/2 = 2$ and $\Phi_{t+1} (c) = 1$ as $h_{t+1}(c)=k-1=1$, we have
$\Delta \Phi (c) \geq 2 - 1 = 1$.
On the other hand, we see that 
$\Delta \Phi (t+1) = \Phi_t (t+1) = \max\{0, 3-k\} = 1$ as $h_t(t+1)=0$.
Moreover, since $g_t(c)\geq 2$, we have $\Delta \Phi(b_1)\geq 1$ by~\eqref{eq:delta_b1}.
Hence, we obtain $\Delta \Phi (b_1)+\Delta \Phi (c)+ \Delta \Phi (t+1) \geq 3$.

\item Consider the case when $k=1$.
Then, since $\Phi_t (c)= k^2/2 +\max\{0, 3-v_t(c)\}$ and $\Phi_{t+1} (c) = 1$ as $h_{t+1}(c)\leq 1$, we have
\[
\Delta \Phi (c) = \frac{k^2}{2}  +\max\{0, 3-v_t(c)\}- 1
= -\frac{1}{2}+\max\{0, 4-g_t(c)\},
\]
since $v_t (c)= g_t(c)-h_t(c)=g_t(c)-k$ and $k=1$.
On the other hand, we see that 
$\Delta \Phi (t+1) = \Phi_t (t+1) = \max\{0, 3-k\} = 2$.
Therefore, we have 
\begin{align*}
\Delta \Phi(b_1) + \Delta \Phi(c)+\Delta \Phi(t+1)\geq 
2+\left(-\frac{1}{2}+\max\{0, 4-g_t(c)\}\right) +\Delta \Phi(b_1) 
\geq \frac{7}{2},
\end{align*}
where the last inequality follows from~\eqref{eq:delta_b1}.
\end{enumerate}
Therefore, we have $\sum_{x\in O_t}(\Phi_t(x) - \Phi_{t+1}(x))\geq 3$ in each case, which implies the lemma.

\subparagraph*{When $s_t\geq 5$~(Figure~\ref{fig:c_and_bot}).}
        \begin{figure}[ht]
            \centering
            \includegraphics[scale=0.3]{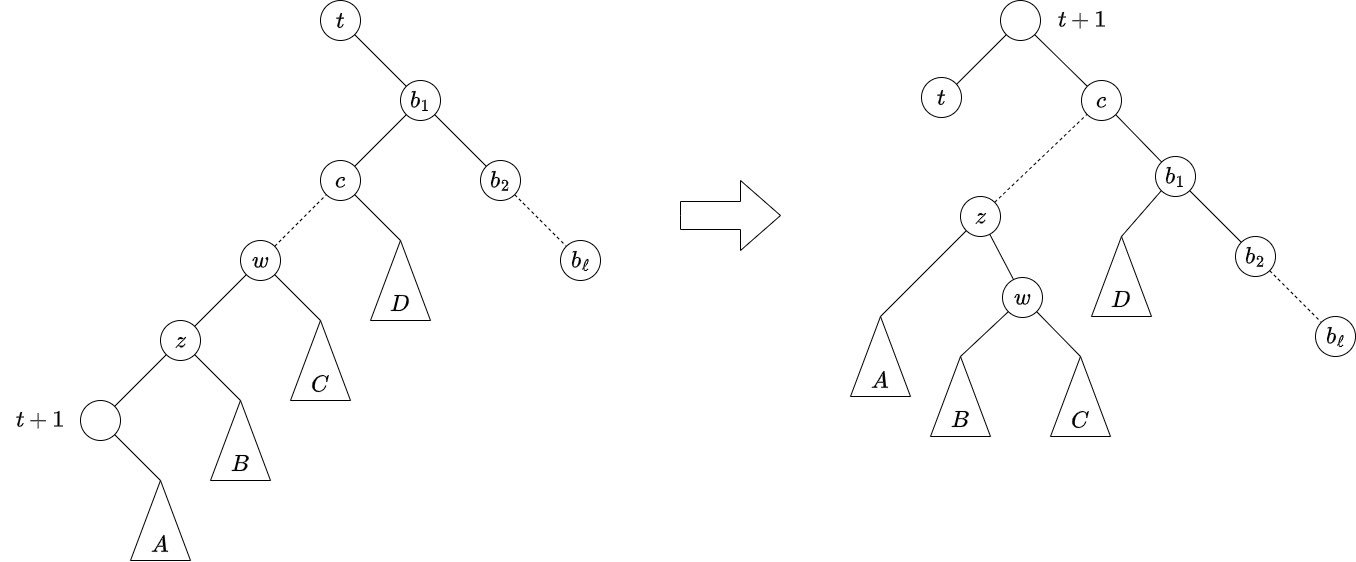}
            \caption{When $s_t\geq 5$.}
            \label{fig:c_and_bot}
        \end{figure}
    There exists a top pair, denoted by $(c, b_1)$, and, after splaying the vertex $t+1$, the vertex $c$ moves to the top spine.
    Also, there exists a bottom pair, denoted by $(z, w)$.
    That is, the left child of $z$ is $t+1$.
    
    Since $\Delta\Phi(x) \geq 0$ for any $x \in O_t$ with $x\not\in\{b_1,c,w,z\}$, we have 
\[
\sum_{x\in O_t}(\Phi_t(x) - \Phi_{t+1}(x))\geq \Delta \Phi (b_1)+\Delta \Phi (c)+\Delta \Phi (w)+\Delta \Phi (z)+ \Delta \Phi (t+1).
\]
        By Corollary~\ref{cor:bottom_poT_rem}, we have 
%\[
$\Delta \Phi(w) + \Delta \Phi(z) \geq h_t(z)+\frac{5}{2}$,
%\]
noting that $2h_t(z)+2\geq h_t(z)+\frac{5}{2}$ since $h_t(z)\geq 1$.
Moreover, since $g_t(t+1) \geq 1$ and $v_t(t+1) = g_t(t+1)$, we have $\Delta\Phi(t+1) = \max\{0, 3-g_t(t+1)\} = \max\{0, 3-h_t(z)\}$.
        Therefore, 
        \[    \Delta\Phi(w)+\Delta\Phi(z) + \Delta\Phi(t+1) \geq \left(h_t(z)+ \frac{5}{2}\right)+\max\{0, 3-h_t(z)\}   \geq \frac{11}{2},
        \]
        and thus        \begin{equation}\label{eq:atleast5}
        \sum_{x\in O_t}(\Phi_t(x) - \Phi_{t+1}(x)) 
        %\geq \Delta \Phi(b_1) + \Delta \Phi(c)+\Delta \Phi(w)+\Delta \Phi(z)+\Delta \Phi(t+1)
        \geq \frac{11}{2} +\Delta \Phi(b_1)+ \Delta \Phi(c).
        \end{equation}
        
    We now evaluate $\Delta \Phi(c)$.
    Letting $k = h_t(c) \geq 1$, we have that $\Phi_t(c) = \frac{k^2}{2}+\max\{0, 3-v_t(c)\}$.
    Since $v_t(c)= g_t(c)-k$, 
    we obtain 
    \[
    \Phi_t(c) = \frac{k^2}{2}+\max\{0, 3+k-g_t(c)\}\geq \frac{k^2}{2} +\max\{0, 4-g_t(c)\}.
    \]
        Since $h_{t+1}(c)=k+1\geq 2$ by Lemma~\ref{lem:refined_h}, we have $\Phi_{t+1}(c) = 2(k+1)-2=2k$.
        Hence, we obtain 
        \begin{align*}
        \Delta \Phi(c) &\geq \frac{k^2}{2}+\max\{0, 4-g_t(c)\} - 2k
        \geq - 2+\max\{0, 4-g_t(c)\}
        \end{align*}
        for $k\geq 1$.
        
        Since~\eqref{eq:delta_b1} also holds for this case, it holds by~\eqref{eq:atleast5} that 
\begin{align*}
\sum_{x\in O_t}(\Phi_t(x) - \Phi_{t+1}(x))
&\geq \frac{11}{2}+\Delta \Phi(b_1) - 2+\max\{0, 4-g_t(c)\}\\
&\geq 
\frac{7}{2}+\Delta \Phi(b_1) + \max\{0, 4-g_t(c)\}
\geq \frac{11}{2}\geq 5,
\end{align*}        
implying the lemma.
\end{proof}

\begin{lemma}\label{lem:black_rotation_even}
    For any time $t<n$ such that $s_t$ is even, it holds that 
    \[
        \sum_{x\in O_t}(\Phi_t(x) - \Phi_{t+1}(x)) \geq
        \begin{cases}
        2& (\text{if\ } s_t = 2)\\
        4 & (\text{if\ } s_t \geq 4)
        \end{cases}.
    \]
\end{lemma}
\begin{proof}
We remark that, after splaying the vertex $t+1$, the top spine remains unchanged.

            \subparagraph*{When $s_t = 2$~(Figure~\ref{fig:c=t+1}).}
There is neither a bottom pair nor a top pair.
Hence, for any $x\in O_t$, we have $\Delta \Phi(x)\geq 0$.

        \begin{figure}[ht]
            \centering
            \includegraphics[scale=0.3]{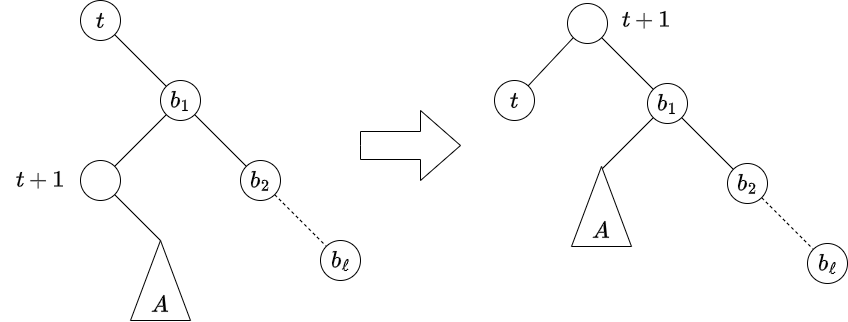}
            \caption{When $s_t = 2$.}
            \label{fig:c=t+1}
        \end{figure}
Let $k=g_t(t+1)\geq 1$.
It holds that $h_t(t+1)=0$, and hence $\Delta \Phi (t+1)= \Phi_t (t+1) = \max\{0, 3-k\}$.
Moreover, we have
\begin{equation}\label{eq:deltab1_even}
\Delta \Phi (b_1)=
\begin{cases}
0 & (k = 1)\\
1 & (k = 2)\\
2 & (k\geq 3 )
\end{cases}
\end{equation}
Indeed, if $k\geq 2$, we have $h_t(b_1) = k$ and $h_{t+1}(b_1)=k-1$, and otherwise, i.e., if $k=1$, then we have $h_{t+1}(b_1)=1$ as the right child of $t+1$ in $T_t$ becomes colored.
This implies~\eqref{eq:deltab1_even}.

Therefore, it holds that
\[
\sum_{x\in O_t}(\Phi_t(x) - \Phi_{t+1}(x)) \geq 
\Delta \Phi (t+1) + \Delta \Phi (b_1)
\geq 
\max\{0, 3-k\} + \Delta \Phi (b_1)
\geq 
2.
\]

\subparagraph*{When $s_t \geq 4$~(Figure~\ref{fig:c_neq_t+1}).}
There exists a bottom pair, denoted by $(z, w)$.
For any $x \in O_t$ with $x\not\in\{b_1, z, w\}$, we have $\Delta\Phi(x) \geq 0$.
Hence, it holds that 
\[
\sum_{x\in O_t}(\Phi_t(x) - \Phi_{t+1}(x)) \geq 
\Delta \Phi (b_1)
+\Delta \Phi (w)
+\Delta \Phi (z)
+\Delta \Phi (t+1)+\Delta \Phi (t+2).
\]

                \begin{figure}[ht]
            \centering
            \includegraphics[scale=0.3]{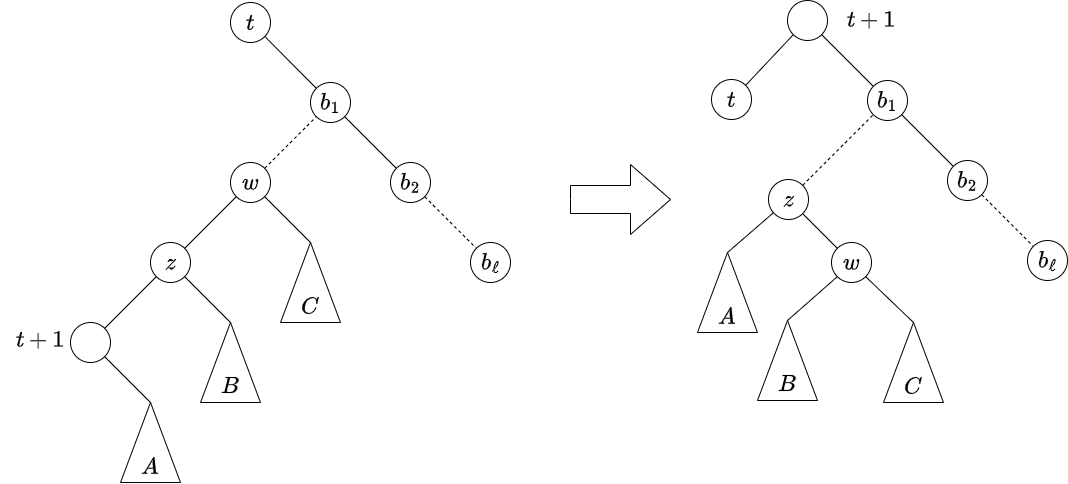}
            \caption{When $s_t \geq 4$.}
            \label{fig:c_neq_t+1}
        \end{figure}

        In what follows, we will show that
\begin{equation}\label{eq:evenfour}     \Delta\Phi(w)+\Delta\Phi(z) + \Delta\Phi(t+1) + \Delta\Phi(t+2)\geq  6.
        \end{equation}
Since $g_t(t+1) \geq 1$ and $h_t(t+1) = 0$, we have $\Delta\Phi(t+1) = \max\{0, 3-g_t(t+1)\} = \max\{0, 3-h_t(z)\}$.

        Suppose that $g_t(t+1)=1$ and the vertex $t+1$ has a right child $y'$.
        Then the vertex $t+2$ is the leftmost leaf of the right sub-tree of $t+1$, which is uncolored at time $t$, as $g_t(t+1)=1$.
        Since the vertex $t+2$ moves to the splaying spine after splaying $t+1$, we have $g_{t+1}(t+2) = 1$, and hence $\Delta \Phi (t+2)=5.5 - 2 = 7/2$.
        By Corollary~\ref{cor:bottom_poT_rem},       
        \[            \Delta\Phi(w)+\Delta\Phi(z) + \Delta\Phi(t+1) + \Delta\Phi(t+2)\geq \left(\frac{5}{2} + h_t(z)\right)+\max\{0, 3-h_t(z)\} +\frac{7}{2}\geq 6.
        \]
        Next consider the other case, that is, when either $g_t(t+1)\geq 2$ or $g_t(t+1)=1$ but $t+1$ has no right child.
        Then, by Corollary~\ref{cor:bottom_poT_rem}, with $\Delta \Phi (t+2)\geq 0$,   
        \[            \Delta\Phi(w)+\Delta\Phi(z) + \Delta\Phi(t+1) + \Delta\Phi(t+2)\geq \left(2 h_t(z)+2\right)+\max\{0, 3-h_t(z)\}   \geq 6,
        \]
        since $h_t(z)\geq 1$.
        Thus~\eqref{eq:evenfour} holds.

Since $h_{t+1}(b_1) = h_t(b_1)+1$, it holds that $\Delta\Phi(b_1) \geq -2$.
Therefore, we have 
\begin{align*}
\sum_{x\in O_t}(\Phi_t(x) - \Phi_{t+1}(x))
&\geq 
\Delta\Phi(b_1)+\Delta\Phi(w) + \Delta\Phi(z) + \Delta\Phi(t+1) + \Delta\Phi(t+2)\\
&\geq -2+6\geq 4.
\end{align*}
%\end{description}
This completes the proof.
\end{proof}

Lemma~\ref{lem:black_rotation} follows from Lemmas~\ref{lem:black_rotation_odd} and~\ref{lem:black_rotation_even}.
%See the full version~\cite{} for the details.
%See Appendix~\ref{sec:appendix} for the details.

\begin{proof}[Proof of Lemma~\ref{lem:black_rotation}]
Recall that the number $r_{t+1}$ of rotations during $\splay{t+1}$ is equal to $s_t$.
Hence it holds that 
\[
r_{t+1} = 
\begin{cases}
2|P_t|+1 & \text{(if $s_t$ is odd)}\\
2|P_t|+2 & \text{(if $s_t$ is even)}
\end{cases},
\]
where we recall that $P_t$ is the set of pairs at time $t$.
By definition, 
we obtain
\[
|P_t|
=
\begin{cases}
0 & \text{(if $s_t\leq 2$)}\\
1 & \text{(if $s_t= 3$)}\\
|M_t|+2 & \text{(if $s_t$ is odd and $s_t\geq 5$)}\\
|M_t|+1 & \text{(if $s_t$ is even and $s_t\geq 4$)}
\end{cases}.
\]
Therefore, if $s_t$ is odd, then
\[
r_{t+1} 
= 
\begin{cases}
1 & \text{(if $s_t=1$)}\\
3 & \text{(if $s_t= 3$)}\\
2 |M_t|+5 & \text{(if $s_t\geq 5$)}
\end{cases},
\]
and, if $s_t$ is even, then
\[
r_{t+1} 
= 
\begin{cases}
2 & \text{(if $s_t=2$)}\\
2 |M_t|+4 & \text{(if $s_t\geq 4$)}
\end{cases}.
\]
Since $|M_t|=0$ if $s_t\leq 3$, 
we conclude from Lemmas~\ref{lem:black_rotation_odd} and~\ref{lem:black_rotation_even} that
\[
r_{t+1} - 2|M_t| \leq  \sum_{x\in O_t}(\Phi_t(x) - \Phi_{t+1}(x))
\]
which completes the proof of Lemma~\ref{lem:black_rotation}.
%\qed
\end{proof}

\section{Lower Bound}\label{sec:LB}

In this section, we show that the number of rotations is at least $(4-o(1))n$ for a left skewed binary tree $T$ of $n$ vertices.

\begin{theorem}
    There exists a splay tree $T$ with $n$ vertices such that 
    $\rot(T)\geq (4-o(1))n$.
\end{theorem}
\begin{proof}
Let $T$ be a binary tree with vertices $1,2,\dots, n$ such that $i$ is the left-child of $i+1$ for $i\geq 1$.
Similarly to Section~\ref{sec:pre}, we add the vertex $0$ to $T$ so that the vertex $n$ is the right-child of $0$.
Then the splaying spine of $T_0$ has vertices $1,2,\dots, n$ in this order from the leaf.
For each $t = 1, 2, \dots, n$, we denote by
$T_t$ the splay tree just after invoking $\splay{t}$ to $T_{t-1}$.
Let $s_t$ denote the number of vertices on the splaying spine of $T_t$.
Recall that we have $\rot (T_0)=\sum_{t=0}^{n-1} s_t$.

We first observe that, for $t<n$, the length $s_t$ of the splaying spine becomes at least $\lfloor s_t/2\rfloor$ after splaying the vertex $t+1$, that is, $s_{t+1}\geq s_{t+1}/2 -1$.
Hence, since $s_0=n$, it holds that
\[
\sum_{t=0}^{\lfloor\log n\rfloor} s_t \geq n + \frac{n}{2}+ \frac{n}{4}+\dots +\frac{n}{2^{\log n}} -\log n= 2n-\log n=(2-o(1))n.
\]

We next evaluate the number of rotations for splaying the vertex $t$ for $t=\lceil\log n\rceil,\dots, n$.
Let $k$ be the total number of vertices that have moved to the top spine during the sequential access.
Note that, once a vertex moves to the top spine, then it remains on the top spine until it is splayed.

We see that, if the vertex $t$ lies on the top spine, splaying $t$ requires $1$ rotation, and otherwise, it takes at least $2$ rotations.
Moreover, if splaying $t$ moves some vertex onto the top spine, it takes at least $3$ rotations.
During the remaining $n-\lceil\log n\rceil$ splaying, at least $k$ splaying are invoked to vertices on the top spine, and at least $k-\lceil\log n\rceil$ splaying moves a vertex onto the top spine.
Therefore, 
the number of rotations after time $\lceil\log n\rceil$ is at least
\[
1\cdot k + 2\left(n-\lceil\log n\rceil - k -(k-\lceil\log n\rceil)\right) + 3\cdot (k-\lceil\log n\rceil) = (2-o(1))n.
\]
Therefore, the total number of rotations is at least $(4-o(1))n$.
\end{proof}

%\clearpage
\bibliography{myrefs}

\end{document}